\documentclass[twoside]{mfa}

\usepackage{amssymb,amsfonts}
\renewcommand\volume{{\bf 1} (2018)}        

\usepackage[active]{srcltx}
\usepackage{graphicx}

\begin{document}

\title[Predicting football tables]{Predicting football tables by a maximally parsimonious model}

\author[Haugen]{Kjetil K. Haugen}
\address{Faculty of Logistics,  Molde University College, Specialized University in Logistics, Britveien 2, 6410 Molde, Norway}
\email{kjetil.haugen@himolde.no}

\author[Owren]{Brynjulf Owren}
\address{Department of Mathematical Sciences, NTNU, N-7034 Trondheim, Norway}
\email{brynjulf.owren@ntnu.no}

\thanks{Thanks to associate professors Halvard Arntzen and Knut P. Heen at Molde University College for useful discussion in preliminary phases of this project. A big thanks also to Lars Aarhus for running the RSSSF Norwegian Football Archive ({\tt www.rsssf.no}) providing running round by round tables in the Norwegian football leagues. All data used in the empirical parts of this article, alongside some useful Fortran 90 programs for football table data management, are availble from the authors upon request.}


{}      

\begin{abstract}
This paper presents some useful mathematical results involved in football table prediction. In addition, some empirical results indicate that an alternative methodology for football table prediction may produce high quality forecasts with far less resource usage than conventional methods.
\end{abstract}

\newtheorem{theorem}{Theorem}[section]
\newtheorem{corollary}[theorem]{Corollary}
\newtheorem{lemma}[theorem]{Lemma}
\newtheorem{proposition}[theorem]{Proposition}

\theoremstyle{definition}
\newtheorem{definition}[theorem]{Definition}
\newtheorem{problem}[theorem]{Problem}
\newtheorem{example}[theorem]{Example}
\newtheorem{remark}[theorem]{Remark}

\numberwithin{equation}{section}

\maketitle
\sloppy

\section{Introduction}\label{sec:Int}

Former England international, long-time Arsenal player and present SKYSPORTS commentator Paul Merson, predicts the final Premiel League (PL) table for the 2016$/$2017 season in~\cite{Merson}. With the hindsight of time, we can check Merson's predictions compared with the true final table as indicated in table~\ref{tab:table_1}.

\begin{table}[htpb]
\begin{center}
\caption{Paul Merson's predictions compared to true final Premier League table.}
\label{tab:table_1}
\begin{tabular}{ll|c|c} 
& \em{Final PL-table} & \em{Merson's predictions} & $|P(i)-i|$ \\
\hline
1.  & Chelsea	   		& 1	&	0 \\
2. & Tottenham	        & 6	&	4 \\
3. & Manchester City 	& 2	&	1 \\
4. & Liverpool 			& 5	&	1 \\
5. & Arsenal 			& 4	&	1 \\
6. & Manchester United	& 3	&	3 \\
7. & Everton			& 8	&	1 \\
8. & Southampton 		& 9	&	1 \\
9. & Bournemouth		& 18	&	9 \\
10. & West Bromwich	& 17	&	7 \\
11.& West Ham		& 7	&	4 \\
12. & Leicester 		& 11	&	1 \\
13. & Stoke	 		& 10	&	3 \\
14. & Crystal Palace		& 12	&	2 \\
15. & Swansea		 	& 15	&	0 \\
16. & Burnley		 	& 20	&	4 \\
17. & Watford		 	& 16	&	1 \\
18. & Hull				& 19	&	1 \\
19. & Middlesbrough 	& 13	&	6 \\
20. & Sunderland	 	& 14	&	6 \\
\end{tabular}
\end{center}
\end{table}

In table~\ref{tab:table_1}, the final table outcome is given in the leftmost column ({\em Final PL-table}), while Merson's predictions are given in the mid column ({\em Merson's predictions}). By defining the true (correct) final table as the consecutive integers $\{1, 2, \ldots, n\footnote{$n$ is the number of teams in the league.}\}$, and a table prediction as a certain permutation $P(i)$\footnote{In this example, $P(i)$ denotes Merson's predictions.} of the integers $\{1, 2, \ldots, n\}$, the absolute deviations between forecasts and true values can be computed as in the rightmost column in table~\ref{tab:table_1}.

If we examine Merson's tips closer, we observe that he obtained two zeros (or perfect hits) in the rightmost column in table~\ref{tab:table_1} -- Chelsea as the winner, and Swansea as number 15. Furthermore, he only missed by one placement for 8 outcomes, but also missed with greater margin for instance for Bournemouth which he thought should finish at 18th place, but in fact ended 9th.

The question that we will be interested in initially, is the quality of Merson's predictions. Is Merson's permutation in table~\ref{tab:table_1} a good guess? In order to attempt to answer such a question, we need to define quality. it seems reasonable to look for some function;

\begin{equation}
f(|P(1)-1|, |P(2)-2|, \ldots, |P(n)-n|)
\end{equation}

which produces a single numerical value tailored for comparison. Of course, infinite possibilities exist for such a function. Fortunately, forecasting literature comes to rescue -- refer for instance to~\cite{Makr}. The two most common measures used in similar situations are $MAE$, Mean Absolute Error or $MSE$, Mean Squared Error. With our notation, these two measures are defined as:

\begin{equation}\label{eq:MAE}
MAE = \frac{1}{n} \sum_{i=1}^n |P(i)-i| \mbox{,   } MSE = \frac{1}{n}  \sum_{i=1}^n  \left( P(i)-i \right)^2
\end{equation}

Although $MSE$ is more applied in statistics, preferably due to its obvious nicer mathematical properties\footnote{Strictly convex for instance.}, we choose to use MAE. It weighs errors equally, and it  produces also an easily interpretable result; a MAE of 3 means that a prediction on average misplaces all teams by 3 places.

Given this choice, Merson's $MAE$ in table~\ref{tab:table_1} can be easily calculated as; $MAE=2.8$. Still, we are in no position to give any statements on the quality of this $MAE$ of 2.8. 

One obvious alternative way to answer our initial question, would be to gather information on other predictions, internet is indeed full of them~(\cite{FT},\cite{Forbes},~\cite{BBC}), calculate $MAE$ for these guesses and compare. Unfortunately, this is indeed a formidable task. As a consequence, we have chosen a slightly different path. Instead of empirical comparisons, we can investigate some basic statistical properties\footnote{Under an assumption of a random guess.} of MAE; for instance to establish minimal ($MAE_{MIN}$), maximal ($MAE_{MAX}$) as well as the expected value for MAE ($E[MAE]$) as a simpler (or at least less time consuming) way of testing the quality of Merson's predictions. It turns out that\footnote{Refer to Appendix~\ref{sec:appa} for the derivation of these results as well as some other relevant statistical properties of $MAE$.} $MAE_{MIN}=0$, $MAE_{MAX}=\frac{n}{2}$ and $E[MAE]= \frac{1}{3} \cdot \frac{n^2-1}{n}$.

The above results provide interesting information. A random table permutation (or prediction) for PL ($n=20$) can at best produce $MAE=0$, while at worst, it can produce $MAE=\frac{n}{2}= \frac{20}{2}=10$. On average, a random prediction should produce $E[MAE]=\frac{1}{3} \cdot \frac{n^2-1}{n} = \frac{1}{3} \cdot \frac{20^2-1}{20}=6.65$.

The task of guessing randomly and hit the correct table is definitely a formidable one. There are $n!$ different tables to guess, and a random guess would hence (in the case of PL, $n=20$)  have a probability of $\frac{1}{20!} = \frac{1}{2432902008176640000} \approx 4 \cdot 10^{-19}$ of hitting the correct table. Consequently, Merson's table prediction is really impressive. On average, a $MAE$ of 6.65 compared to Merson's 2.8 indicate high quality in Merson's prediction. One could of course argue that Paul Merson is an expert, and one should expect him to know this business\footnote{Some might even argue that PL is a league with low uncertainty of outcome (a competitively imbalanced league). Hence, it is not that hard to guess final tables. Chelsea, Arsenal and Manchester United have for instance a recurring tendency to end up among the 5 best.}. Still, information from other countries, for instance Norway, which we will focus more on in subsequent sections, indicate that even experts may have challenges in providing tips that fit final tables.

In the next section (section~\ref{sec:lit}), we investigate some scientific attempts to produce football table forecasts. In section~\ref{sec:fcast} we argue that the trend in present research seems to be oriented in a non-parsimonious fashion, and argue why this perhaps is not a good idea. In section~\ref{sec:hyp} , we discuss alternative parsimonious modelling hypotheses and test one involving goal-difference as the prime explanatory factor. Section~\ref{sec:conc} concludes, and discusses and suggests further research.

\section{The science of football table prediction}\label{sec:lit}
Although the internet is ``full'' of football table predcitions, it would be an exaggeration to state that research literature is full of serious attempts to predict the same. Still, some noteworthy exceptions exist. Three relatively recent papers by Brillinger~\cite{Brill1},~\cite{Brill2},~\cite{Brill3} seem to sum up state of the art of the area. Brillinger's touch of difference compared to other previous work seems to be that he models game outcomes in the form of Win Tie or Loss -- W, T, L -- directly, as opposed to other authors who uses some distributional assumptions on goal scoring frequencey, typically as seen in~\cite{Lee} or in~\cite{Meeden}\footnote{Meeden's paper does get some interesting and perhaps unexpected criticism in~\cite{HaugenChance}. Here, the whole assumption of using probability theory to model goal scoring or match outcomes is questioned by game theoretic arguments.}.

Almost all of the work discussed above rely on simulation to produce actual forecasts. The idea is simple. Let the computer play the games; either by drawing goal scores  or W, T, L (by estimated probailistic mechanisms) for all predefined matches in the league. Register match outcomes; either by counting goals or more directly by Brillingers approach. Then, when all match outcomes are defined, the league table can be set-up. Repeating the simulation produces a new final league table, and by a large number of simulation runs, expected table placement or probabilistic table predictions can be generated. Of course, such a method opens up for updating or reestimating underlying probabilties for team quality, which then are applied if one runs a rolling horizon approach. Such rolling horizon approaches seem to be quite popular in media -- refer for instance to~\cite{Sig}.

\section{Parsimonious forecasting}\label{sec:fcast}

The concept of parsimony  (or parameter minimzation) is both well known and well studied in time series forecasting literature. Already Box and Jenkins~\cite{Box} pointed out that parsimony is desireable if forecasting accuracy is the objective. An interesting emprical test of the actual consequences of parsimony versus non-parsimomny can be found in~\cite{Ledolter}.

The reason why parsimony is desirable is obvious. A model where many parameters need to be estimated generate more  aggregate uncertainty than a model with fewer parameters. As a consequence, the outcome -- the forecasts -- tends to be more uncertain and inaccurate.

Furthermore, in many cases where causal (regression type models) are used, either alone or in combination with time series models, the causal variables will often have to be predicted in order to obtain model estimates for the target variable. And, these causal variables are typically just as hard, or (perhaps) even harder to predict reasonably correct, than the target variable. Suppose you want to predict the number of flats sold in a certain area in London this month next year. You know that many relevant economic variables like UK salary level,  unemployment rate and Interest rates (just to name few) affect this target variable. If a prediction model contains these variables, you need to predict them in order to predict the number of flats sold in London next year. And, predicting next years UK unemployment, salary as well as interest rates are (obviously) not an easy task.

If we return back to our focus -- football table prediction -- it should seem quite obvious that the reported methodolgy discussed in section~\ref{sec:lit}, hardly can be described as parsimonious. On the contrary, probability estimates of many teams, maybe conditional on future events like injuries, or talent logistics shoud generate much added uncertainty and it should not come as a surprise that such methods produce quite bad predictions. The fact that the team in~\cite{Sig} missed Greece as a potential winner in EURO 2004 may serve as an adequate example.

This said, non-parsimonious models, either causal or not, have other interesting properties, they can for instance (far better) answer questions of 'what if type', which in some situations are more desirable than accurate forecasts.

So, what would be a parsimonious model for football table forecasting? The answer is simple and obvious, the table itself. Either last years table, if one predicts the final table in-between seasons, or the latest table available if the aim is to predict the final table within a rolling horizon. 

Obviously, even such a simple strategy holds challenges. In most leagues there are relegation and promotion which has the obvious effect that last season's table contains a few other teams than this season's table. Furthermore, if the tables that are to be predicted are group tables in say European or World Championships, there is no previous season. 

Still, such problems may be solved at least if we restrict our focus to prediction after some games or rounds are played. 

\section{Testing a hypothesis of parsimonious football table prediction}\label{sec:hyp}

One simple way of testing the table in a certain round $r$'s predictive power on the final table is to perform a set of linear regressions, one for each round with the final table rank as the dependent variable and table rank in round $r$ as the independent variable. Or, in our notation: (the obvious $r$-subscript is omitted for simplicity)

\begin{equation}
i = \beta_0 + \beta_1 P(i) + \epsilon_i
\end{equation}

By calculating $R^2$ in all these regressions, a function $R^2(r)$ is obtained. Presumably, this function will have some kind of increasing pattern (not necessarily strict), but common sense indicates that football tables change less in later than early rounds. Figure shows an example from last years Tippeliga\footnote{The name Tippeligaen has been changed to Eliteserien for this (2017) season.} in Norway.

\begin{figure}[htpb]
\centerline{\includegraphics[scale=0.40]{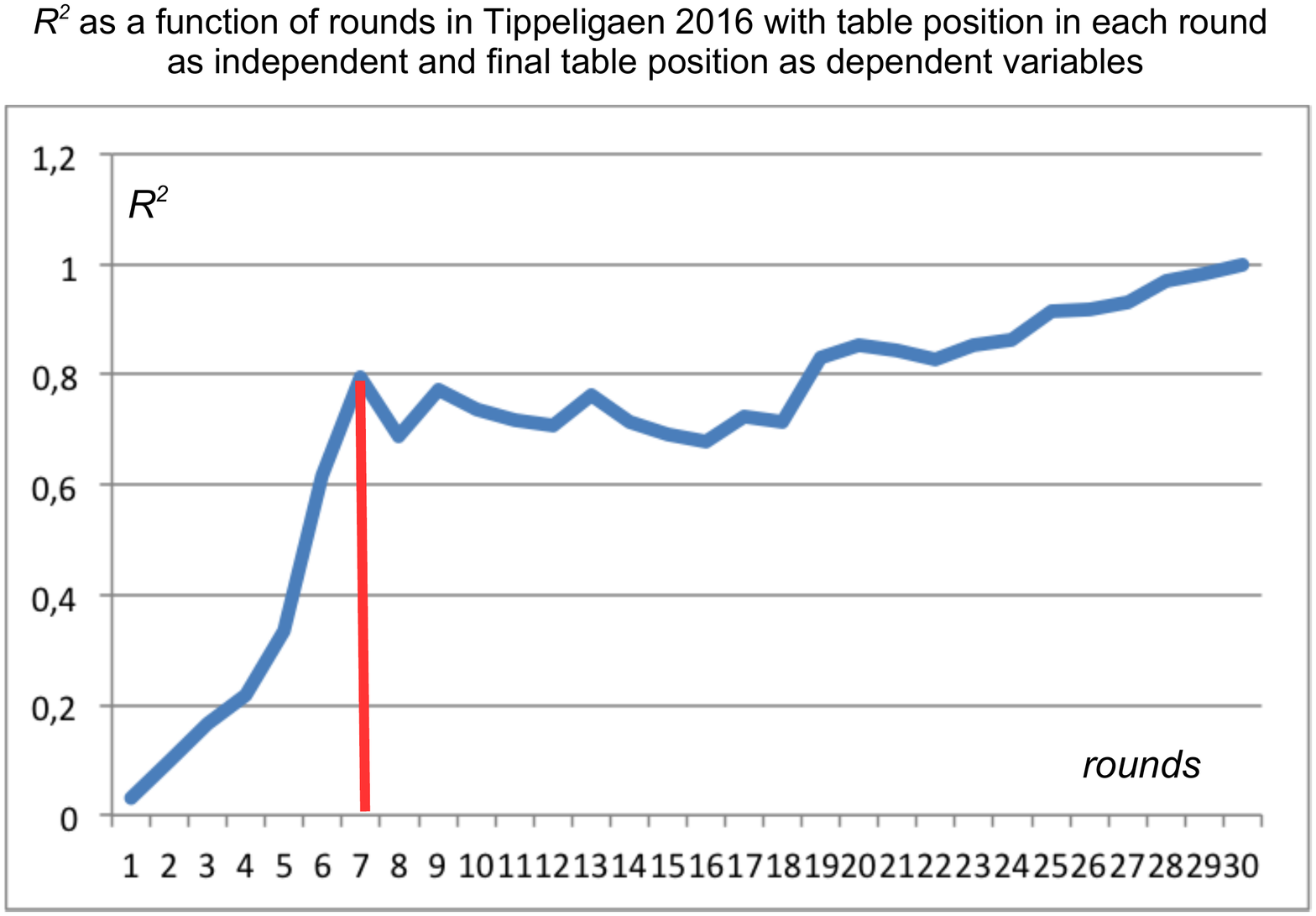}}
\caption{An example of $R^2(r)$ from Tippeligaen 2016}
\label{fig:fig_1}
\end{figure}

Looking at figure~\ref{fig:fig_1}, we observe our predicted pattern of non-strict positive monotonicity in $R^2(r)$. However, we also observe something else: $R^2(r)$ reaches 80\% explanatory power already in round 7. That is, 80\% of the final table is there, already in round 7. Surely, $R^2(r)$ drops slightly in rounds 7 to 19, but this observation indicates that our parsimonious hypothesis actually may be of relevance.

Now, is the table rank the only possible parsimonious alternative? The answer is of course no. A table contains home wins, away wins, points, goal-score to name some potential additional information. Let us focus on goal score. In the start of a season, many teams may have new players, new managers and we may suspect that the full potential of certain good teams may not be revealed in early table rankings. Vice versa, other teams have luck, are riding a wave and take more points than expected.  These, not so good teams, may have a tendency to win even matches by a single goal, but also loose other matches (against very good teams) by many goals. As such, we could suspect that goal difference (at least in earlier rounds) perhaps could hold more and better predictive information than table ranking (or points for that matter). That is, we could hope to observe patterns (of course not as smooth) similar to the ``fish-form'' in figure~\ref{fig:fig_2}.

\begin{figure}[htpb]
\centerline{\includegraphics[scale=0.40]{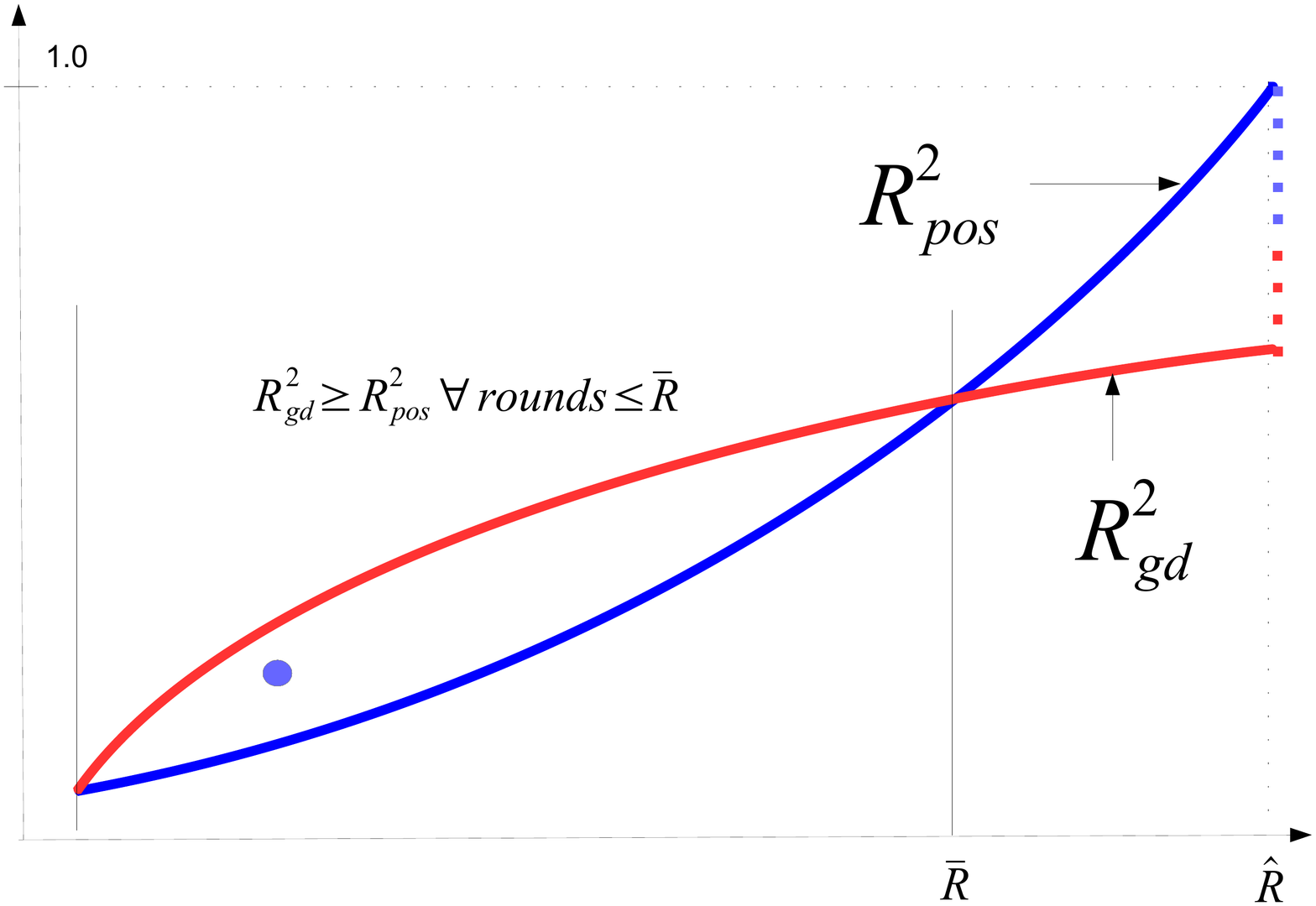}}
\caption{A hypothetical comparison of $R^2(r)$ for regressions with both goal difference ($R^2_{gd}$) and table position ($R^2_{pos}$)as independent variables}
\label{fig:fig_2}
\end{figure}

In order to check these two hypotheses more thoroughly, We examined the Norwegian top league for some previous seasons. We stopped in 2009, as the number of teams changed to 16 (from 14) this year, and performed regressions like those described above with both table rank and goal difference as independent variables. Then 8 $R^2_{pos}(r)$ and 8 $R^2_{gd}(r)$ were generated. The result of this generation is shown in figure~\ref{fig:fig_3}

\begin{figure}[htpb]
\centerline{\includegraphics[scale=0.70]{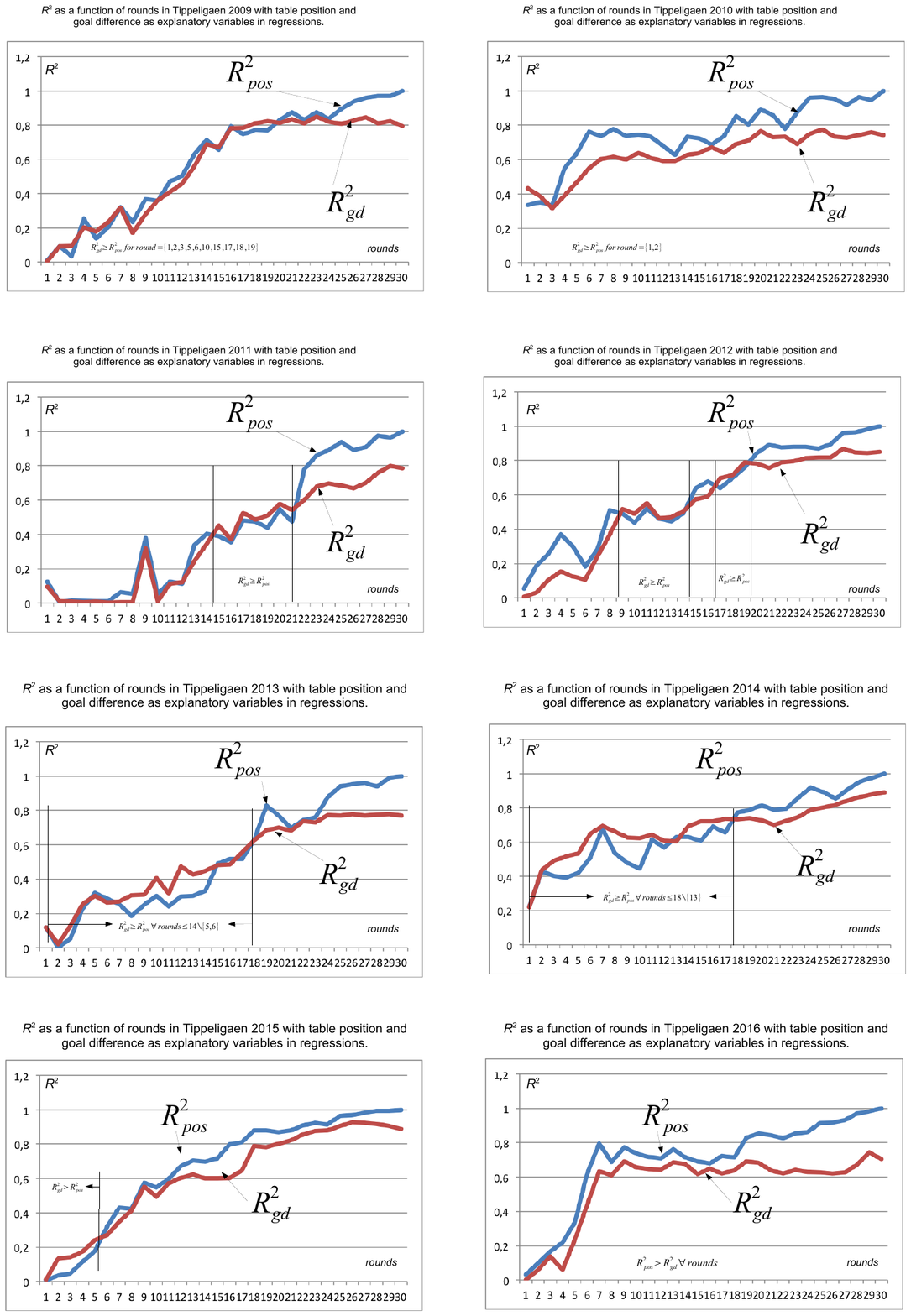}}
\caption{Full output from empirical analysis}
\label{fig:fig_3}
\end{figure}

If we examine figure~\ref{fig:fig_3} we observe that 7 out of 8 seasons have patterns, although perhaps not visually very similar to figure~\ref{fig:fig_2}), where goal difference explain final table better than table rank -- typically early in the season. Furthermore, around 80\% explanatory power ($R^2 > 0.8$) is obtained roughly (for table rank) as early as mid-season for most cases.

Hence, an operative predictive strategy where a table prediction simply could be generated by sorting goal differences in early parts of the season and using the last observed table as the forecast later in the season seems reasonable.

Hence, we have demonstrated that our hypotheses are supported for these Norwegian data. Our results do of course not state anything related other leagues in other countries, but we do believe that similar patters should be observable in most international leagues. 

\section{Conclusions and suggestions for further research}\label{sec:conc}

Apart from the fact that Paul Merson is a good predictor of the final PL-table (or at least he was before the 2016/2017 season), we have demonstrated that table rank or sometimes even better, goal difference explains major parts of final tables early. We have not (actually) checked (empirically) if our prediction method is ``better'' than existing methods in research literature. This is of course feasible, however time consuming. What we without doubt can conclude, is that the methods applied by researchers in football table prediction are far more time and resource consuming than our methods. Simulation experiments take both coding and computing time, and if one is in doubt about whether one approach is better than the other, at least we could recommend to try our approach first.

So, is football table prediction important? Does it contribute to world welfare? Is it really necessary to spend time and resources even addressing this problem? Perhaps not. Still, most modern news agents spend a lot of time and (valuable) space on distributing such predictions each season. And,  as a consequence, some real world demand seems to exist.

In any case, we have examined the problem from our perspective and even found some mathematical results we found interesting. Hopefully, our small effort may inspire other researchers to start the tedious job of empirically testing whether our approach performs better or worse than the simulation-based approaches.

\appendix

\section{Statistical properties of $MAE$}\label{sec:appa}

\subsection{Maximal and minimal values for $MAE$}\label{subsec:proof}
The minimal value of $MAE$ ($MAE_{MIN}$) is obvious. Even though it is unlikely (refer to section~\ref{sec:Int}),  it is possible to guess correctly. In that case, $P(i)=i \forall i$, and by equations~(\ref{eq:MAE}), $MAE=0$.

Let us proceed by investigating the maximal $MAE$. We start by introducing $S(P)$ as:

\begin{equation}\label{eq:abs}
S(P) = \sum_{i=1}^n |P(i)-i|
\end{equation}

then,

\begin{equation}\label{eq:MAEavP}
MAE(P) = \frac{1}{n} S(P)
\end{equation}

Our focus is on solving the optimization problem: $\max_P MAE(P)$. As equation~(\ref{eq:MAEavP}) indicates only a multiplied constant difference, we might as well focus on solving:

\begin{equation}\label{eq:opt}
\max_P S(P)
\end{equation}

\begin{theorem} 
The permutation $P=P_0=[n, n-1, \ldots 2,1]$ will be one (although not necessarily unique) solution to~(\ref{eq:opt}), and the value of the optimal objective (given an even numbered league\footnote{By obvious reasons, we restrict ourselves to investigating even numbered leagues. That is, $n=2m$}) is $\frac{n}{2}$.
\end{theorem}

\begin{proof} Suppose an alternative permutation, say $P_1$ such that $S(P_1) > S(P_0)$ exists. We will show that such a permutation, $P_1$ can not exist. In order to proceed, we introduce a little trick to handle the absolute value in equation~(\ref{eq:abs}). let us illustrate the trick by an example.

$$
 \begin{array}{|c|c|}
 \hline
 P(1) & 1 \\ \hline
 P(2) & 2 \\ \hline
 \vdots & \vdots \\ \hline
 P(n) & n \\ \hline
 \end{array}
 \longrightarrow
 \begin{array}{|c|c|}
 \hline
 \alpha_1 & \beta_1 \\ \hline
 \alpha_2 & \beta_2 \\ \hline
 \vdots & \vdots \\ \hline
 \alpha_n & \beta_n \\ \hline
 \end{array}
 \qquad\text{Example}\quad
 \begin{array}{|c|c|}
 \hline
 6 & 1 \\ \hline
 1 & 2 \\ \hline
 4 & 3 \\ \hline
  5 & 4 \\ \hline
  2 & 5 \\ \hline
  3 & 6 \\ \hline
 \end{array}
 \longrightarrow
 \begin{array}{|c|c|}
 \hline
 6 & 1 \\ \hline
 2 & 1 \\ \hline
 4 & 3  \\ \hline
 5 &  4 \\ \hline
 5 & 2  \\ \hline
 6 & 3 \\ \hline
 \end{array}
 $$

Above (on the left), a certain prediction (or permutation) $[P(1), P(2), \ldots$, P(n)] alongside the correct final table $[1,2, \ldots n]$ is given. The trick involves a certain resort of this leftmost table into the table with $\alpha$'s and $\beta$'s, and is done as follows: $\alpha_i = max\{P(i),i\}$ and $\beta_i = min\{P(i),i\}$. The clue of the trick (the resort) is of course to achieve that $\alpha_i > \beta_i \forall i$, remove the absolute value sign, and hence obtain:

\begin{equation}
S(P) = \sum_{i=1}^n |P(i)-i| = \sum_{i=1}^n \left( \alpha_i -\beta_i \right) = \sum_{i=1}^n \alpha_i - \sum_{i=1}^n \beta_i
\end{equation}

Now, this reformulation leads to a simpler task in solving $\max_P S(P)$, as it can be done by maximising $\sum_{i=1}^m \alpha_i$ and minimising $\sum_{I=1}^n \beta_i$. The ``best'' that can be achieved is to get the column of $\alpha_i$'s to contain two copies of each of the numbers $m+1, \ldots,n$ such that the $\beta_i$-column contains two copies  of the numbers $1, \ldots,m$ (where $m=\frac{n}{2}$). Permutations $P$ that makes this happen are those where $P(i) \geq m+1$ for $i \leq m$ and $P(i) \leq m$ for $i \geq m+1$. That is, collectively, $P$ must ``ship'' the set $\{1,\ldots,m\}$ to $\{m+1,\ldots,n\}$ and vice versa. Or, seen from the predictor's side, if the best half of the teams are predicted to end up on the bottom half of the table (and vice versa): Then, we get:

\begin{equation}
S^*(P) = 2 \sum_{j=1}^m (m+j) - 2\sum_{j=1}^m j = 2m^2 = \frac{1}{2}n^2
\end{equation}

or $MAE= n S^*(P) = \frac{n}{2}$. Luckily, the permutation $P_0$ satisfies this criteria. Reversing the order, predicting the winner to be last, number two to become second last and so on, will with necessity secure that the best $m$ are ``shipped'' to the last $m$ and vice versa.

\end{proof}



\subsection{Deriving the expression for $E[MAE]$}

To derive the expression for $E[MAE]$, we assume uniform distribution of all possible tables (random guess). 
Let $x=(x_1,\ldots,x_n)=(P(1),\ldots,P(n))$ be a permutation. We write $S(x_1,\ldots,d_n)$ for the score, and $p(x_1,\ldots,x_n)=1/n!$ for the probability that $x$ occurs.

The expected value of $S$ is then (by definition):

\begin{align*}
\mathbb{E}[S] &= \sum_{x\in\mathcal{S}_n} S(x_1,\ldots,x_n)p(x_1,\ldots,x_n)
= \sum_{i=1}^n  \sum_{x\in\mathcal{S}_n} |x_i - i| p(x_1,\ldots,x_n)\\
&= \sum_{i=1}^n \sum_{x_i=1}^n |x_i-i| \sum_{y\in\mathcal{S}_{n-1}} p(y_1,\ldots,y_{i-1},x_i,y_{i},\ldots,y_{n-1})
\end{align*}

In the innermost (right) summation, we sum over all permutations$\{1,\ldots,n\}\backslash \{x_i\}$. That is, we leave out $x_i$ fom the numbers in all permutations, $\{1,\ldots,x_i-1,x_{i}+1,\ldots,n\}$. Then, the innermost summation contains all permuations of $(n-1)$ numbers. Hence, there are $(n-1)!$ of them. We get:

\begin{align*}
\mathbb{E}[S]&=\sum_{i=1}^n \sum_{x_i=1}^n |x_i-i| \frac{(n-1)!}{n!}
=\frac1{n}\sum_{i=1}^n\left(\sum_{x_i=1}^{i-1}(i-x_i)+\sum_{x_i=i+1}^{n}(x_i-i)\right)\\
&= \frac1n\sum_{i=1}^n (f(i)+f(n+1-i))=\frac2n\sum_{i=1}^n f(i),\quad\text{hvor}\ f(x)=\frac12 x(x-1)
\end{align*}

The last summation is straightforward to find, and we get:

\begin{equation}
\mathbb{E}[S]=\frac1n\sum_{i=1}^n i(i-1)=\frac13 (n^2-1) 
\end{equation}

and hence $E(MAE] =  \frac{1}{3} \cdot \frac{n^2-1}{n}$. Similarly, we can also find the variance:

\begin{equation}
\text{Var}[S] = \frac{1}{45} (n+1)(2n^2+7),\quad 
\text{Var}[\text{MAE}] = \frac{1}{45} \frac{(n+1)(2n^2+7)}{n^2}.\quad 
\end{equation}




\subsection{Some final mathematical remarks}

If we revert back to the final paragraph of subsection~\ref{subsec:proof}, it should to be reasonably straightforward to realize that the mentioned criteria (for achieving maximal $MAE$) can be achieved in exactly $(m!)^2$ ways. (Take all $M!$ permutations of $[1,\ldots,m]$ and add $m$, for any such permutation, take all $m!$ permutations and subtract $m$.) Since there are $n! = (2m)!$ permutations totally, the probability of obtaining the worst possible guess can be calculated as:

\begin{equation}
P\left( \mbox{Getting $MAE_{MAX}$ by guessing} \right) = \frac{(m!)^2}{(2m)!} = \left(\begin{array}{c} 2m\\ m\end{array}\right)^{-1} = \left(\begin{array}{c} n\\ n/2\end{array}\right)^{-1}
\end{equation}

This outcome is unlikely. In the PL-case: 

\begin{equation}
\left(\begin{array}{c} n\\ n/2\end{array}\right)^{-1} = \left(\begin{array}{c} 20\\ 10 \end{array}\right)^{-1} = \frac{1}{184756} \approx 5.4 \cdot 10^{-6}
\end{equation}

Still, far from as unlikely as guessing the correct table, as indicated in section~\ref{sec:Int}.

\bibliographystyle{plain}

\end{document}